\documentclass[lettersize,journal]{IEEEtran}

\usepackage{booktabs}
\usepackage{amsmath,amsfonts,amsthm,amssymb}
\usepackage{algorithmic}
\usepackage{algorithm}
\usepackage{array}
\usepackage[caption=false,font=normalsize,labelfont=sf,textfont=sf]{subfig}
\usepackage{stfloats}
\usepackage{url}
\usepackage{verbatim}
\usepackage{graphicx}
\usepackage[noadjust]{cite}
\usepackage{comment}

\allowdisplaybreaks

\theoremstyle{remark}
\newtheorem{example}{\indent Example}
\newtheorem{definition}{\indent Definition}
\newtheorem{lemma}{\indent Lemma}
\newtheorem{corollary}{\indent Corollary}

\newtheorem{theorem}{\indent Theorem}
\newtheorem*{remark}{\indent Remark}
\newtheorem*{conjecture}{\indent Conjecture}

\begin{document}

\title{The Weight Distributions of Two Classes of Linear Codes From Perfect Nonlinear Functions}

\author{Huawei~Wu,~Jing~Yang~and~Keqin~Feng
  %~\IEEEmembership{}
  % <-this % stops a space
  %\thanks{K. Feng is with the Department of Mathematical Sciences, Tsinghua University, Beijing, 100084, China (e-mail: fengkq@mail.tsinghua.edu.cn)}% <-this % stops a space
  \thanks{H. Wu is with the Department of Mathematical Sciences, Tsinghua University, Beijing, 100084, China (e-mail: wu-hw18@mails.tsinghua.edu.cn)}%
  \thanks{J. Yang is with the Department of Mathematical Sciences, Tsinghua University, Beijing, 100084, China (e-mail: y-j@mail.tsinghua.edu.cn)}%
  \thanks{K. Feng is with the Department of Mathematical Sciences, Tsinghua University, Beijing, 100084, China (e-mail: fengkq@mail.tsinghua.edu.cn)}
  \thanks{The work of K. Feng is supported by the National Natural Science Foundation of China under Grant 12031011.}
}

\markboth{}%
{Shell \MakeLowercase{\textit{et al.}}: A Sample Article Using IEEEtran.cls for IEEE Journals}


\maketitle

\begin{abstract}
  In this paper, we employ general results on the value distributions of perfect nonlinear functions from $\mathbb{F}_{p^m}$ to $\mathbb{F}_p$  to give a unified approach to determining the weight distributions of two classes of linear codes over $\mathbb{F}_p$ constructed from perfect nonlinear functions, where $p$ is an odd prime and $m$ is an odd number. When $m$ is even, we give some mild additional conditions for similar conclusions to hold.
\end{abstract}


\section{Introduction}
\IEEEPARstart{L}{et} $A$ and $B$ be two finite abelian groups and let $f:A\rightarrow B$ be a mapping. We say that $f$ is perfect nonlinear if for any $a\in A\backslash\{0\}$, the difference function $D_af:A\rightarrow B$ given by $x\mapsto f(x+a)-f(x)$ is balanced, i.e., $|D_af^{-1}(b)|=|A|/|B|$ for any $b\in B$. If $|A|=|B|$, perfect nonlinear mappings from $A$ to $B$ are also called planar functions. Perfect nonlinearity is a robust measure of nonlinearity related to differential cryptanalysis, which was first introduced by Nyberg in \cite{nyberg1991perfect}. We refer to \cite{carlet2004highly} for an intensive survey on perfect nonlinear functions (and, more generally, highly nonlinear functions).\\
\indent Throughout this section, let $p$ be an odd prime number and let $q=p^m$ with $m\in\mathbb{N}_+$, where $\mathbb{N}_+$ is the set of positive integers. We are mainly interested in perfect nonlinear functions from $\mathbb{F}_q$ to itself. By the above definition, a function $f:\mathbb{F}_q\rightarrow\mathbb{F}_q$ is perfect nonlinear if for any $a\in\mathbb{F}_q^*$, the difference function $D_af:\mathbb{F}_q\rightarrow\mathbb{F}_q$ is bijective.

\begin{comment}
\indent Perfect nonlinear functions over $\mathbb{F}_q$ are usually classified up to a certain notion of equivalence. Two functions $f$, $g:\mathbb{F}_q\rightarrow\mathbb{F}_q$ are called extended affine equivalent if there exist two affine permutations $l_1$, $l_2:\mathbb{F}_q\rightarrow\mathbb{F}_q$ and an affine function $l_3:\mathbb{F}_q\rightarrow\mathbb{F}_q$ such that $g=l_1\circ f\circ l_2+l_3$, and are called affine (resp., linear) equivalent if there exist two affine (resp., linear) functions $l_1$, $l_2:\mathbb{F}_q\rightarrow\mathbb{F}_q$ such that $g=l_1\circ g\circ l_2$. The thesis [4] serves as a nice introduction to these equivalence relations.
\end{comment}

\indent It is quite difficult to construct new perfect nonlinear functions. Up to now, all known perfect nonlinear functions over $\mathbb{F}_q$ are affine equivalent to one of the following functions:
\begin{enumerate}
  \item Dembowski-Ostrom type:
        $$\Pi_1(x)=\sum\limits_{0\le i\le j\le m-1}a_{ij}x^{p^i+p^j},$$
        where $a_{ij}\in\mathbb{F}_{q}$ satisfy specific restrictions (see \cite{budaghyan2008new}, \cite{coulter2007planar}, \cite{dembowski1968planes}, \cite{ding2006family}, \cite{zha2009perfect});
  \item Coulter-Matthews type: $\Pi_2(x)=x^{\frac{3^k+1}{2}}$, where $p=3$, $k$ is odd and $\gcd(m,k)=1$ (see \cite{coulter1997planar}).
\end{enumerate}

\indent Perfect nonlinear functions can also be used to construct linear codes with good parameters. Let $\Pi:\mathbb{F}_{q}\rightarrow\mathbb{F}_q$ be a perfect nonlinear function. Carlet, Ding and Yuan studied in \cite{carlet2005linear} the following two classes of linear codes over $\mathbb{F}_{p}$:
$$C_{\Pi}=\{c_{a,b}=(f_{a,b}(x))_{x\in\mathbb{F}_q^{*}}:\ a,b\in\mathbb{F}_{q}\},$$
where
$$f_{a,b}(x)={\rm{Tr}}_{\mathbb{F}_q/\mathbb{F}_{p}}(a\Pi(x)+bx),$$
and
$$\overline{C_{\Pi}}=  \{c_{a,b,c}=(f_{a,b,c}(x))_{x\in\mathbb{F}_q}:  \ a,b, c\in\mathbb{F}_{q}\},$$
where
$$f_{a,b,c}(x)={\rm{Tr}}_{\mathbb{F}_q/\mathbb{F}_{p}}(a\Pi(x)+bx+c).$$
They showed that $C_{\Pi}$ is a $[q-1,2m]_p$ linear code if $\Pi(0)=0$ and $\overline{C_{\Pi}}$ is a $[q,1+2m]_p$ linear code, and developed bounds on the nonzero Hamming weights of $C_{\Pi}$ and $\overline{C_{\Pi}}$. They also studied the dual codes of $C_{\Pi}$ and $\overline{C_{\Pi}}$, proving results on the minimum distances.\\
\indent In \cite{yuan2006weight}, Carlet, Ding and Yuan determined the weight distribution of $C_{\Pi}$ when $\Pi$ is a perfect nonlinear function of one of the following forms:
\begin{enumerate}
  \item $\Pi(x)=x^{p^k+1}$, where $k\in\mathbb{N}$ and $m/\gcd(m,k)$ is an odd number;
  \item $\Pi(x)=x^{\frac{3^k+1}{2}}$, where $p=3$, $k$, $m$ are odd and $\gcd(m,k)=1$;
  \item $\Pi(x)=x^{10}-ux^6-u^2x^2$, where $p=3$, $m$ is odd and $u\in\mathbb{F}_q^*$.
\end{enumerate}
The authors did that on a case-by-case basis and left the Coulter-Matthews type of perfect nonlinear functions, where $m$ is even, as an open problem. They also demonstrated that the linear codes $C_{\Pi}$ constructed from the above three types of perfect nonlinear functions are either optimal or among the best codes known. Here, an $[n,k,d]$ linear code over $\mathbb{F}_p$ is called optimal if there does not exist an $[n,k,d']$ linear code over $\mathbb{F}_p$ such that $d'>d$. The term "best code" can be understood in a similar manner.\\
\indent In \cite{feng2007value}, Feng and Luo calculated the value distributions of certain exponential sums derived from perfect nonlinear functions and utilized their results to establish a unified method of determining the weight distribution of $C_{\Pi}$ when $\Pi$ is a perfect nonlinear function of Dembowski-Ostrom type or Coulter-Matthews type.\ Up to then, the weight distributions of the linear codes $C_{\Pi}$ were completely determined for all known perfect nonlinear functions.\\
\indent In \cite{li2008covering}, Li, Ling and Qu employed a unified approach to determining the weight distributions of $C_{\Pi}$ and $\overline{C_{\Pi}}$ for all the perfect nonlinear functions in the list treated in \cite{yuan2006weight}. They were the first to determine the weight distributions of the linear codes $\overline{C_{\Pi}}$ and their method of determining the weight distributions of $C_{\Pi}$ was new.\\
\indent The main tools of the previous studies mentioned above were exponential sums and quadratic forms over finite fields. The main drawback of these studies is that they depend on the specific form of the perfect nonlinear function $\Pi$ and thus they were all done on a case-by-case basis; in particular, their conclusions hold only for known perfect nonlinear functions.\\
\indent However, as shown in \cite{yuan2006weight}, \cite{feng2007value} and \cite{li2008covering}, the results on the weight distributions of $C_{\Pi}$ and $\overline{C_{\Pi}}$ are the same for all known perfect nonlinear functions, respectively. Hence, it is natural to ask whether these results hold for an arbitrary perfect nonlinear function and whether we can find a unified approach to proving them rather than doing that on a case-by-case basis.\\
\indent In \cite{li2008properties}, when $p=3$, only using the assumption that $\Pi$ has perfect nonlinearity, the authors adopted a new approach to determining the weight distribution of $C_{\Pi}$ by determining all the possibilities for the value distribution of a perfect nonlinear function from $\mathbb{F}_{3^m}$ to $\mathbb{F}_3$. Their starting point was the following lemma concerning the value distribution of a perfect nonlinear mapping.
\begin{lemma}[{\cite[Theorem 9]{carlet2004highly}}]
  Let $(A,+)$ and $(B,+)$ be abelian groups of order $n$ and $m$, respectively, where $m$ divides $n$, and let $f:A\rightarrow B$ be a perfect nonlinear mapping. For any $b\in B$, put $k_b=|f^{-1}(b)|$. Then we have
  \begin{equation}\label{20230531equationone}
    \begin{cases}
      \sum\limits_{z\in B}k_z^2=\frac{n^2+(m-1)n}{m},                                   \\
      \sum\limits_{z\in B}k_zk_{z+b}=\frac{n(n-1)}{m},\ \forall\ b\in B\backslash\{0\}, \\
      \sum\limits_{z\in B}k_z=n.
    \end{cases}
  \end{equation}
\end{lemma}
\noindent In general, the equations in the second row of (\ref{20230531equationone}) are not symmetric, causing difficulty in solving them. However, if  $A=\mathbb{F}_{3^m}$ and $B=\mathbb{F}_3$, then the equations in (\ref{20230531equationone}) are simple enough (in particular, they are symmetric) for us to determine all the solutions by converting the problem into finding all the integer representations of $3^{m-1}$ by the binary quadratic form
$$X^2+XY+Y^2.$$
This strategy cannot be generalized to arbitrary $p$ since the equations in (\ref{20230531equationone}) are difficult to solve.\\
\indent In this paper, we employ general results on the value distributions of perfect nonlinear functions from $\mathbb{F}_{q}$ to $\mathbb{F}_p$ to determine the weight distributions of $C_{\Pi}$ and $\overline{C_{\Pi}}$ only using the assumption that $\Pi$ has perfect nonlinearity for all odd primes $p$ when $m$ is odd. When $m$ is even, we give some mild additional conditions for similar conclusions to hold. Moreover, for the linear code $\overline{C_{\Pi}}$, we can say more about its codewords than just determining the weight distribution from the perspective of a group action. This paper not only provides results that go beyond the literature,
but also presents some deep insights on these two classes of linear codes.

\section{The Value Distributions of Perfect Nonlinear Functions}

\indent There is a characterization of perfect nonlinearity by means of the Fourier transform on finite abelian groups. For this, let us first recall some basic definitions and results.\\
\indent Let $G$ be a finite abelian group and let $f:G\rightarrow\mathbb{C}$ be a complex-valued function. The Fourier transform of $f$ is defined by the complex-valued function
\begin{align*}
  \hat{f}:\widehat{G} & \rightarrow\mathbb{C},                  \\
  \chi                & \mapsto\sum\limits_{g\in G}f(g)\chi(g),
\end{align*}
where $\widehat{G}$ is the group of characters of $G$, which is (in general, non-canonically) isomorphic to $G$.\ If $F:A\rightarrow B$ is a mapping between two finite abelian groups, then for any $\chi\in\widehat{B}$, we have the complex-valued function $F_{\chi}:=\chi\circ F:A\rightarrow\mathbb{C}$. We define the Walsh transform $W_F:\widehat{A}\times\widehat{B}\rightarrow\mathbb{C}$ of $F$ by
$$W_F(\phi,\chi)=\widehat{F_{\chi}}(\phi).$$
The set
$$\{W_F(\phi,\chi):\ \phi\in\widehat{A},\ \chi\in\widehat{B}\backslash\{1_{\widehat{B}}\}\}$$
is called the Walsh spectrum of $F$, where $1_{\widehat{B}}$ is the trivial character of $B$ sending all elements of $B$ to $1\in\mathbb{C}$.\\
\indent The following theorem gives a characterization of perfect nonlinearity using the Walsh spectrum.

\begin{theorem}[{\cite[Theorem 16]{carlet2004highly}}]\label{20230602themone}
  Let $F:A\rightarrow B$ be a mapping between two finite abelian groups. Then $F$ is perfect nonlinear if and only if $|W_F(\phi,\chi)|=\sqrt{|A|}$ for any $\phi\in\widehat{A}$ and $\chi\in\widehat{B}\backslash\{1_{\widehat{B}}\}$.
\end{theorem}

\indent With the abstract theory out of the way, let us now focus on the finite field case.\ Let $p$ be a prime number and let $q=p^m$ with $m\in\mathbb{N}_+$.\\
\indent If we fix a primitive $p$-th root of unity $\xi_p$ in $\mathbb{C}$, then we have a natural isomorphism $\mathbb{F}_q\rightarrow\widehat{\mathbb{F}_q}$, $a\mapsto\psi_a$, where $\psi_a$ is given by
\begin{align*}
  \psi_a:\mathbb{F}_q & \rightarrow\mathbb{C},                                      \\
  x                   & \mapsto\xi_p^{\mathrm{Tr}_{\mathbb{F}_q/\mathbb{F}_p}(ax)}.
\end{align*}
Therefore, if $f$ is a complex-valued function on $\mathbb{F}_q$, we usually define its Fourier transform by
\begin{align*}
  \hat{f}:\mathbb{F}_q & \rightarrow\mathbb{C}                                                                       \\
  a                    & \mapsto\sum\limits_{x\in\mathbb{F}_q}f(x)\xi_p^{{\rm{Tr}}_{\mathbb{F}_q/\mathbb{F}_p}(ax)},
\end{align*}
which is a complex-valued function on $\mathbb{F}_q$ rather than one on $\widehat{\mathbb{F}_q}$.\ Let $h$ be a positive divisor of $m$. The Walsh transform of a function $F:\mathbb{F}_{q}\rightarrow\mathbb{F}_{p^h}$ is similarly defined by
\begin{align*}
  W_F(a,b) & =\sum\limits_{x\in\mathbb{F}_{q}}\psi_a(x)\psi_b(F(x))                                                                                 \\
           & =\sum\limits_{x\in\mathbb{F}_{q}}\xi_p^{{\rm{Tr}}_{\mathbb{F}_{p^h}/\mathbb{F}_p}(bF(x))+{\rm{Tr}}_{\mathbb{F}_{q}/\mathbb{F}_p}(ax)},
\end{align*}
with $a\in\mathbb{F}_{q}$ and $b\in\mathbb{F}_{p^h}$. Note that, in the literature, $W_F(a,b)$ is usually defined by
$$W_F(a,b)=\sum\limits_{x\in\mathbb{F}_{q}}\xi_p^{{\rm{Tr}}_{\mathbb{F}_{p^h}/\mathbb{F}_p}(bF(x))-{\rm{Tr}}_{\mathbb{F}_{q}/\mathbb{F}_p}(ax)}.$$
However, the Walsh spectra corresponding to both definitions are identical. By Theorem \ref{20230602themone}, $F$ is perfect nonlinear if and only if $|W_F(a,b)|=\sqrt{q}$ for any $a\in\mathbb{F}_{q}$ and $b\in\mathbb{F}_{p^h}^*$.\\
\indent Recall that a function $f:\mathbb{F}_q\rightarrow\mathbb{F}_p$ is called bent if $W_f(a):=W_f(a,1)$ has absolute value $\sqrt{q}$ for any $a\in\mathbb{F}_q$. The following lemma shows the equivalence between perfect nonlinearity and bentness.

\begin{lemma}[{\cite[Theorem 2.3]{nyberg1991constructions}}]
  Let $f:\mathbb{F}_q\rightarrow\mathbb{F}_p$ be a function. Then $f$ is perfect nonlinear if and only if it is bent.
\end{lemma}

\indent This characterization of perfect nonlinear functions allows us to study their value distributions. Let $f:\mathbb{F}_q\rightarrow\mathbb{F}_p$ be a function. We say that the value distribution of $f$ is $(n_0,\cdots,n_{p-1})$ if $n_i=|f^{-1}(i)|$ for any $0\le i\le p-1$. The following two theorems describe completely the value distribution of a perfect nonlinear function from $\mathbb{F}_{q}$ to $\mathbb{F}_p$.

\begin{theorem}[{\cite[Theorem 3.2]{nyberg1991constructions}}]\label{20230602themfour}
  Assume that $m$ is even. Then for any bent function $f:\mathbb{F}_{q}\rightarrow\mathbb{F}_p$, there exists $s\in\{0,1,\cdots,p-1\}$ such that the value distribution of $f$  is $(n_0,n_1,\cdots,n_{p-1})$, where
  \begin{align}
     & n_s=p^{m-1}+(p-1)p^{\frac{m}{2}-1},           \label{20230603equationnine} \\
     & n_i=p^{m-1}-p^{\frac{m}{2}-1},\quad i\ne s,\notag
  \end{align}
  or
  \begin{align}
     & n_s=p^{m-1}-(p-1)p^{\frac{m}{2}-1},    \label{20230603equationten} \\
     & n_i=p^{m-1}+p^{\frac{m}{2}-1},\quad i\ne s.\notag
  \end{align}
\end{theorem}

\begin{theorem}[{\cite[Theorem 3.4]{nyberg1991constructions}}]\label{20230602themtwo}
  Assume that $p$ and $m$ are both odd. Then for any bent function $f:\mathbb{F}_q\rightarrow\mathbb{F}_p$, there exists $s\in\{0,1,\cdots,p-1\}$ such that the value distribution of $f$ is  $(n_0,n_1,\cdots,n_{p-1})$, where
  \begin{align*}
    n_i=p^{m-1}+(\frac{i+s}{p})p^{\frac{m-1}{2}},\quad i=0,\ \cdots,\ p-1,
  \end{align*}
  or
  \begin{align*}
    n_i=p^{m-1}-(\frac{i+s}{p})p^{\frac{m-1}{2}},\quad i=0,\ \cdots,\ p-1.
  \end{align*}
  Here, $(\frac{\cdot}{p})$ is the Legendre symbol modulo $p$ and we establish the convention that $(\frac{0}{p})=0$.
\end{theorem}
\begin{remark}
  Note that, in \cite{nyberg1991constructions}, Theorem \ref{20230602themtwo} was stated only for regular bent functions (see Definition \ref{20230604defone}). However, the proof can be easily generalized to show the validity of Theorem \ref{20230602themtwo} for  arbitrary bent functions.
\end{remark}

\indent The key points for proving Theorem \ref{20230602themfour} and Theorem \ref{20230602themtwo} are the following two lemmas.

\begin{lemma}[{\cite[Lemma before Theorem 3.4]{nyberg1991constructions}}]\label{20230602lemmathree}
  Assume that $p$ is odd. If there exist $a_1$, $\cdots$, $a_{p-1}\in\mathbb{Q}$ such that
  \begin{align*}
        & a_1\xi_p+a_2\xi_p^2+\cdots+a_{p-1}\xi_p^{p-1}                                   \\
    =\  & \begin{cases}
            \sqrt{p},  & \text{if $p\equiv1\ ({\rm{mod}}\ 4)$}, \\
            i\sqrt{p}, & \text{if $p\equiv3\ ({\rm{mod}}\ 4)$},
          \end{cases}
  \end{align*}
  then $a_i=(\frac{i}{p})$ for $i=1$, $\cdots$, $p-1$.
\end{lemma}

\begin{comment}
\begin{proof}
  By calculating quadratic Gaussian sums (see, e.g., \cite[Theorem 5.15]{lidl1997finite}), we can see that $a_i=(\frac{i}{p})$ ($1\le i\le p-1$) is a solution. Moreover, since the dimension of $\mathbb{Q}(\xi_p)$ over $\mathbb{Q}$ is $p-1$, $\xi_p$, $\cdots$, $\xi_p^{p-1}$ form a $\mathbb{Q}$-basis for $\mathbb{Q}(\xi_p)$. It follows that the solution is unique.
\end{proof}
\end{comment}

\begin{lemma}[{\cite[Property 7, 8]{kumar1985generalized}}]\label{20230602lemmatwo}
  Assume that $p$ is odd. For any bent function $f:\mathbb{F}_q\rightarrow\mathbb{F}_p$, there exists a function $f^*:\mathbb{F}_q\rightarrow\mathbb{F}_p$, which is called the dual of $f$, such that for any $a\in\mathbb{F}_q$,
  \begin{equation}\label{20230603equationeleven}
    W_f(a)=\begin{cases}
      \pm\xi_p^{f^*(a)}\sqrt{q},   & \mbox{if }p\equiv 1\ ({\rm{mod}}\ 4), \\
      \pm i\xi_p^{f^*(a)}\sqrt{q}, & \mbox{if }p\equiv 3\ ({\rm{mod}}\ 4)
    \end{cases}
  \end{equation}
  if $m$ is odd and
  \begin{equation}\label{20230603equationfourteen}
    W_f(a)=\pm\xi_p^{f^*(a)}\sqrt{q}
  \end{equation}
  if $m$ is even. The sign $\pm 1$ in (\ref{20230603equationeleven}) or (\ref{20230603equationfourteen}) will be called the sign of $f$ at $a$.
\end{lemma}

\indent We provide a sketch of the proofs for Theorem \ref{20230602themfour} and Theorem \ref{20230602themtwo} here, which will be used in the proofs of Corollary \ref{20230603corolone} and Theorem \ref{202363theoremsix}. Indeed, the value distribution of a bent function $f:\mathbb{F}_q\rightarrow\mathbb{F}_p$ is related to its Walsh spectrum via the following equation:
$$
  W_f(0)  =\sum\limits_{x\in\mathbb{F}_q}\xi_p^{f(x)}=\sum\limits_{i=0}^{p-1}n_i\xi_p^i,
$$
where $n_i=|f^{-1}(i)|$ for any $0\le i\le p-1$. If $m$ is even, then by Lemma \ref{20230602lemmatwo}, we have
$$W_f(0)=\epsilon\xi_p^s\sqrt{q},$$
where $\epsilon\in\{\pm 1\}$ is the sign of $f$ at $0$ and $0\le s\le p-1$. By employing some basic properties of cyclotomic fields, we can show that
$$n_s=p^{m-1}+\epsilon(p-1)p^{\frac{m}{2}-1}$$
and
$$n_i=n_s-\epsilon p^{\frac{m}{2}}=p^{m-1}-\epsilon p^{\frac{m}{2}-1}$$
for any $i\ne s$. In particular, the value distribution of $f$ is of the form (\ref{20230603equationnine}) (resp., (\ref{20230603equationten})) if and only if the sign of $f$ at $0$ is $+1$ (resp., $-1$). Theorem \ref{20230602themtwo} can be proved similarly, with the additional assistance of Lemma \ref{20230602lemmathree}.

\indent For the purpose of later exposition of the results, we introduce the definition of weakly regular bentness.

\begin{definition}\label{20230604defone}
  Let $p$ be an odd prime number and let $f:\mathbb{F}_q\rightarrow\mathbb{F}_p$ be a bent function. If the signs $\pm 1$ in (\ref{20230603equationeleven}) or (\ref{20230603equationfourteen}) are the same for all $a\in\mathbb{F}_q$, then $f$ is called weakly regular bent. In this case, the common sign will be called the sign of $f$. If $W_f(a)=\xi_p^{f^*(a)}\sqrt{q}$ for any $a\in\mathbb{F}_q$, where $f^*:\mathbb{F}_q\rightarrow\mathbb{F}_p$ is the dual of $f$, then $f$ is called regular bent.\\
  \indent Let $\Pi:\mathbb{F}_q\rightarrow\mathbb{F}_q$ be a perfect nonlinear function. We say that $\Pi$ is weakly regular perfect nonlinear if the functions $\Pi_a:\mathbb{F}_q\rightarrow\mathbb{F}_p$ given by $x\mapsto{\rm{Tr}}_{\mathbb{F}_q/\mathbb{F}_p}(a\Pi(x))$ are weakly regular bent for all $a\in\mathbb{F}_q^*$.
\end{definition}

\begin{corollary}\label{20230603corolone}
  Assume that $m$ is even and let $f:\mathbb{F}_q\rightarrow\mathbb{F}_p$ be a weakly regular bent function. If the value distribution of $f$ is of the form (\ref{20230603equationnine}) (resp., (\ref{20230603equationten})), then so is that of $f_b$ for any $b\in\mathbb{F}_q$, where
  $$f_{b}(x)=f(x)+{\rm{Tr}}_{\mathbb{F}_{q}/\mathbb{F}_p}(bx).$$
\end{corollary}

\begin{proof}
  As mentioned above, the form of the value distribution of a bent function from $\mathbb{F}_q$ to $\mathbb{F}_p$ depends only on the sign of the function at $0$.  Since $f$ is weakly regular bent, the signs of $f$ at $0$ and $b$ are the same. Moreover, we have
  $$
    W_{f}(b)=\sum\limits_{x\in\mathbb{F}_q}\xi_p^{f(x)+{\rm{Tr}}_{\mathbb{F}_q/\mathbb{F}_p}(bx)}=W_{f_b}(0),
  $$
  which implies that the sign of $f_b$ at $0$ is equal to the sign of $f$ at $b$. Therefore, $f_b$ and $f$ have the same sign at $0$.
\end{proof}

\section{The Weight Distributions of $C_{\Pi}$ and $\overline{C_{\Pi}}$}
\indent Throughout this section, let $p$ be an odd prime number, let $q=p^m$ with $m\in\mathbb{N}_+$ and let $\Pi:\mathbb{F}_q\rightarrow\mathbb{F}_q$ be a perfect nonlinear function. For any $0\le i\le q-1$, let $A_i$ denote the number of codewords of Hamming weight $i$ in $C_{\Pi}$ and for any $0\le i\le q$, let $\overline{A}_i$ denote the number of codewords of Hamming weight $i$ in $\overline{C_{\Pi}}$.\\
\indent We first consider the linear code $\overline{C_{\Pi}}$, since we can say more about its codewords than just determining the weight distribution. For any $c_{a,b,c}\in\overline{C_{\Pi}}$, we say that it is a $(n_0,\cdots,n_{p-1})$-codeword if the value distribution of $f_{a,b,c}$ is $(n_0,\cdots,n_{p-1})$.\\
\indent If $a=0$ and $b=0$, then $f_{0,0,c}(x)={\rm{Tr}}_{\mathbb{F}_q/\mathbb{F}_p}(c)$ for any $x\in\mathbb{F}_q$. For any $0\le i\le p-1$, if ${\rm{Tr}}_{\mathbb{F}_q/\mathbb{F}_p}(c)=i$, then $c_{0,0,c}=(i,\cdots,i)$, i.e., $c_{0,0,c}$ is a $(0,\cdots,0,p^m,0,\cdots,0)$-codeword, where $p^m$ occurs in the $i$-th entry. There is only one codeword of each of these $p$ types in $\overline{C_{\Pi}}$.\\
\indent If $a=0$ but $b\ne 0$, then for any $c\in\mathbb{F}_{q}$, since $x\mapsto bx+c$ is a permutation of $\mathbb{F}_q$, $c_{0,b,c}$ is a $(p^{m-1},\cdots,p^{m-1})$-codeword. There are $p(q-1)=p^{m+1}-p$ such codewords.\\
\indent If $a\ne 0$, then for any $b$, $c\in\mathbb{F}_q$, $f_{a,b,c}$ is a perfect nonlinear function from $\mathbb{F}_q$ to $\mathbb{F}_p$. Put
$$\overline{\Omega}=\{c_{a,b,c}\in\overline{C_{\Pi}}:\ a\ne 0\}.$$
Since for any $a$, $a'$, $b$, $b'$, $c$, $c'\in\mathbb{F}_q$, $c_{a,b,c}=c_{a',b',c'}$ if and only if $a=a'$, $b=b'$ and ${\rm{Tr}}_{\mathbb{F}_q/\mathbb{F}_p}(c-c')=0$, the subset $\overline{\Omega}$ is well-defined, i.e., the condition that $a\ne 0$ does not depend on the way in which the codewords in $\overline{C_{\Pi}}$ are represented as $c_{a,b,c}$. It is easy to see that there are $pq(q-1)=p^{m+1}(p^m-1)$ codewords in $\overline{\Omega}$.\\
\indent Consider the group $G:=\mathbb{F}_p^{*}\times\mathbb{F}_p$, whose multiplication is given by
$$(\alpha_1,\beta_1)\cdot(\alpha_2,\beta_2)=(\alpha_1\alpha_2,\alpha_1\beta_2+\beta_1).$$
The identity element of $G$ is $(1,0)$. For any $\alpha\in\mathbb{F}_p^*$, $\beta\in\mathbb{F}_p$ and $a$, $b$, $c\in\mathbb{F}_q$, we have
\begin{align*}
  \alpha f_{a,b,c}(x)+\beta & =\alpha{\rm{Tr}}_{\mathbb{F}_q/\mathbb{F}_p}(a\Pi(x)+bx+c)+\beta                 \\
                            & ={\rm{Tr}}_{\mathbb{F}_q/\mathbb{F}_p}(\alpha a\Pi(x)+\alpha bx+\alpha c+\beta')
\end{align*}
for any $x\in\mathbb{F}_q$, i.e., $\alpha f_{a,b,c}+\beta=f_{\alpha a,\alpha b,\alpha c+\beta'}$, where $\beta'\in\mathbb{F}_q$ is such that ${\rm{Tr}}_{\mathbb{F}_q/\mathbb{F}_p}(\beta')=\beta$. Therefore, we can define a $G$-action on $\overline{\Omega}$ by
$$(\alpha,\beta)\cdot c_{a,b,c}=c_{\alpha a,\alpha b,\alpha c+\beta'}.$$
(Note that this definition does not depend on the choice of $\beta'$.)

\begin{lemma}
  The $G$-action on $\overline{\Omega}$ is free; that is, if $(\alpha,\beta)\cdot c_{a,b,c}=c_{a,b,c}$ for some $(\alpha,\beta)\in G$ and $c_{a,b,c}\in\overline{\Omega}$, then $(\alpha,\beta)$ must be the identity element of $G$, i.e., $(1,0)$.
\end{lemma}
\begin{proof}
  Assume that $(\alpha,\beta)\in G$ fixes $c_{a,b,c}\in\overline{\Omega}$, i.e.,
  $$(\alpha,\beta)\cdot c_{a,b,c}=c_{\alpha a,\alpha b,\alpha c+\beta'}=c_{a,b,c}.$$
  Then
  $$\alpha a=a,\ \alpha b=b,\ {\rm{Tr}}_{\mathbb{F}_{q}/\mathbb{F}_p}(\alpha c+\beta'-c)=0.$$
  Since $a\ne 0$, we must have $\alpha=1$ and thus
  $$\beta={\rm{Tr}}_{\mathbb{F}_{q}/\mathbb{F}_p}(\beta')=0.$$
  Therefore, $G$ acts freely on $\overline{\Omega}$.
\end{proof}

\indent In particular, there are $p^m(p^m-1)/(p-1)$ orbits of $\overline{\Omega}$ under the $G$-action and each orbit contains $(p-1)p$ codewords. The essence of proving Theorem \ref{202363theoremsix} lies in a careful examination of how the types of codewords in $\overline{\Omega}$ evolve under the $G$-action. It can be easily seen that if $c_{a,b,c}$ is a $(n_0,\cdots,n_{p-1})$-codeword, then for any $(\alpha,\beta)\in G$, $(\alpha,\beta)\cdot c_{a,b,c}$ is a $(n_0',\cdots,n_{p-1}')$-codeword, where $n_i'=n_{\alpha^{-1}(i-\beta)}$ for any $0\le i\le p-1$.\\
\indent For the sake of convenient exposition, we introduce the following definition.

\begin{definition}
  Assume that $m$ is odd. A codeword $c_{a,b,c}\in\overline{\Omega}$ is called an $s^+$-codeword (resp., $s^-$-codeword) with $0\le s\le p-1$ if it is a $(n_0,\cdots,n_{p-1})$-codeword, where
  \begin{align*}
    n_i=p^{m-1}+(\frac{i+s}{p})p^{\frac{m-1}{2}}\ \mbox{for any}\ 0\le i\le p-1 \\
    (\mbox{resp., }n_i=p^{m-1}-(\frac{i+s}{p})p^{\frac{m-1}{2}}\ \mbox{for any}\ 0\le i\le p-1).
  \end{align*}
  \indent Assume that $m$ is even. A codeword $c_{a,b,c}\in\overline{\Omega}$ is called an $s^+$-codeword (resp., $s^-$-codeword) with $0\le s\le p-1$ if it is a $(n_0,\cdots,n_{p-1})$-codeword, where
  \begin{align*}
    n_i=\begin{cases}
          p^{m-1}+(p-1)p^{\frac{m}{2}-1}, & \mbox{if }i=s,    \\
          p^{m-1}-p^{\frac{m}{2}-1},      & \mbox{if }i\ne s,
        \end{cases} \\
    \Big(\mbox{resp., }n_i=\begin{cases}
                             p^{m-1}-(p-1)p^{\frac{m}{2}-1}, & \mbox{if }i=s,    \\
                             p^{m-1}+p^{\frac{m}{2}-1},      & \mbox{if }i\ne s,
                           \end{cases}\Big).
  \end{align*}
  \indent If $c_{a,b,c}\in\overline{\Omega}$ is an $s^+$-codeword (resp., $s^-$-codeword) for some $0\le s\le p-1$, then $c_{a,b,c}$ is called a positive (resp., negative) codeword and is called a strictly positive (resp., negative) codeword if furthermore $s\ne 0$.
\end{definition}

\indent By Theorem \ref{20230602themfour} and Theorem \ref{20230602themtwo}, any codeword $c_{a,b,c}\in\overline{\Omega}$ is an $s^+$-codeword or an $s^-$-codeword for some $0\le s\le p-1$. We can now determine the number of codewords in $\overline{C_{\Pi}}$ of each type, as follows.

\begin{theorem}\label{202363theoremsix}
  Let $\Pi:\mathbb{F}_q\rightarrow\mathbb{F}_q$ be a  perfect nonlinear function.
  \begin{enumerate}
    \item If $m$ is odd, then the possible types of codewords in $\overline{C_{\Pi}}$ and the corresponding numbers of codewords are listed in Table I. As a consequence, $\overline{A_i}=0$ except for the values
          \begin{align*}
             & \overline{A}_0=1,                                                         \\
             & \overline{A}_{(p-1)p^{m-1}-p^{\frac{m-1}{2}}}=\frac{1}{2}(p-1)p^m(p^m-1), \\
             & \overline{A}_{(p-1)p^{m-1}}=(p^m-1)(p^m+p),                               \\
             & \overline{A}_{(p-1)p^{m-1}+p^{\frac{m-1}{2}}}=\frac{1}{2}(p-1)p^m(p^m-1), \\
             & \overline{A}_{p^m}=p-1.
          \end{align*}
          In particular, the minimum distance of $\overline{C_{\Pi}}$ is $(p-1)p^{m-1}-p^{\frac{m-1}{2}}$ and $\overline{C_{\Pi}}$ is a $4$-weight code.
    \item If $m$ is even, then the possible types of codewords in $\overline{C_{\Pi}}$ and the corresponding numbers of codewords are listed in Table I, assuming that one of the following conditions holds:
          \begin{enumerate}
            \item $p=3$,
            \item $|\Pi^{-1}(0)|=1$, $\Pi$ is weakly regular perfect nonlinear (see Definition \ref{20230604defone}) and the map $\mathbb{F}_q^*\rightarrow\mathbb{F}_p$ given by  $a\mapsto\Pi_a^*(0)$ is not surjective, where $\Pi_a=f_{a,0,0}$ and $\Pi_a^*$ is the dual of $\Pi_a$ (see Lemma \ref{20230602lemmatwo}).
          \end{enumerate}
          As a consequence, $\overline{A_i}=0$ except for the values
          \begin{align*}
             & \overline{A}_0=1,                                                        \\
             & \overline{A}_{(p-1)(p^{m-1}-p^{\frac{m}{2}-1})}=\frac{p^m(p^m-1)}{2},    \\
             & \overline{A}_{(p-1)p^{m-1}-p^{\frac{m}{2}-1}}=\frac{p^m(p^m-1)(p-1)}{2}, \\
             & A_{(p-1)p^{m-1}}=p^{m+1}-p,                                              \\
             & \overline{A}_{(p-1)p^{m-1}+p^{\frac{m}{2}-1}}=\frac{p^m(p^m-1)(p-1)}{2}, \\
             & \overline{A}_{(p-1)(p^{m-1}+p^{\frac{m}{2}-1})}=\frac{p^m(p^m-1)}{2},    \\
             & \overline{A}_{p^m}=p-1.
          \end{align*}
          In particular, the minimum distance of $\overline{C_{\Pi}}$ is $(p-1)(p^{m-1}-p^{\frac{m}{2}-1})$ and $\overline{C_{\Pi}}$ is a $6$-weight code.
  \end{enumerate}
  \begin{table}[!ht]\label{20230531tabletwo}
    \centering
    \caption{}
    \renewcommand\arraystretch{2.5}
    \begin{tabular}{cc}
      \toprule %[2pt]    
      Type                              & Number                 \\
      \midrule %[2pt]  
      $\substack{(0,\cdots,0,p^m,0,\cdots,0)                     \\(\forall\ 0\le i\le p-1)}$ & $1$                    \\\hline
      $(p^{m-1},\cdots,p^{m-1})$        & $p^{m+1}-p$
      \\\hline
      $s^+$\ ($\forall\ 0\le s\le p-1$) & $\frac{p^m}{2}(p^m-1)$
      \\\hline
      $s^-$\ ($\forall\ 0\le s\le p-1$) & $\frac{p^m}{2}(p^m-1)$ \\
      \bottomrule %[2pt]     
    \end{tabular}
  \end{table}
\end{theorem}
\begin{proof}
  Assume first that $m$ is odd. Let $(\alpha,\beta)\in G$ and let $c_{a,b,c}\in\overline{\Omega}$. If $c_{a,b,c}$ is an $s^+$-codeword, then $(\alpha,\beta)\cdot c_{a,b,c}$ is a $(\alpha s-\beta)^+$-codeword if $(\frac{\alpha^{-1}}{p})=1$ and is a $(\alpha s-\beta)^-$-codeword if $(\frac{\alpha^{-1}}{p})=-1$. Similarly, if $c_{a,b,c}$ is an $s^-$-codeword, then $(\alpha,\beta)\cdot c_{a,b,c}$ is a $(\alpha s-\beta)^-$-codeword if $(\frac{\alpha^{-1}}{p})=1$ and is a $(\alpha s-\beta)^+$-codeword if $(\frac{\alpha^{-1}}{p})=-1$.\ In particular, every orbit of $\overline{\Omega}$ contains a $0^+$-codeword.\\
  \indent Let $O$ be an orbit of $\overline{\Omega}$ and let $c_{a,b,c}$ be a $0^+$-codeword in $O$. Then for any $(\alpha,\beta)\in G$, $(\alpha,\beta)\cdot c_{a,b,c}$ is a $(-\beta)^+$-codeword if $(\frac{\alpha^{-1}}{p})=1$ and is a $(-\beta)^-$-codeword if $(\frac{\alpha^{-1}}{p})=-1$. In $\mathbb{F}_p^*$, there are $(p-1)/2$ elements $\alpha$ such that $(\frac{\alpha^{-1}}{p})=1$ and $(p-1)/2$ elements $\alpha$ such that $(\frac{\alpha^{-1}}{p})=-1$. Hence, in $O$, there are $(p-1)/2$ $s^+$-codewords and $(p-1)/2$ $s^-$-codewords for any $0\le s\le p-1$.\ This proves the assertion on the possible types of codewords in $\overline{C_{\Pi}}$.\\
  \indent To determine the weight distribution of $\overline{C_{\Pi}}$,\ we only consider the codewords in $\overline{\Omega}$, since the Hamming weights of the other codewords are easy to determine. It is clear that a codeword in $\overline{\Omega}$ has Hamming weight $(p-1)p^{m-1}$ if and only if it is a $0^+$-codeword or a $0^-$-codeword. Let $1\le s\le p-1$. If $(\frac{s}{p})=1$, then any $s^+$-codeword has Hamming weight $(p-1)p^{m-1}-p^{\frac{m-1}{2}}$ and any $s^-$-codeword has Hamming weight $(p-1)p^{m-1}+p^{\frac{m-1}{2}}$. If $(\frac{s}{p})=-1$, the situation is reversed. Hence, in $\overline{\Omega}$ (and thus in $\overline{C_{\Pi}}$), the codewords of Hamming weight $(p-1)p^{m-1}-p^{\frac{m-1}{2}}$ are equal in number to the codewords of Hamming weight $(p-1)p^{m-1}+p^{\frac{m-1}{2}}$. This proves the assertion on the weight distribution of $\overline{C_{\Pi}}$.\\
  \indent Now assume that $m$ is even. Let $(\alpha,\beta)\in G$ and let $c_{a,b,c}\in\overline{\Omega}$. If $c_{a,b,c}$ is an $s^+$-codeword, then $(\alpha,\beta)\cdot c_{a,b,c}$ is a $(\alpha s+\beta)^+$-codeword; in particular, the orbit containing $c_{a,b,c}$ contains a $0^+$-codeword. If $c_{a,b,c}$ is an $s^-$-codeword, then $(\alpha,\beta)\cdot c_{a,b,c}$ is a $(\alpha s+\beta)^-$-codeword; in particular, the orbit containing $c_{a,b,c}$ contains a $0^-$-codeword.\ Let $O$ be an orbit of $\overline{\Omega}$. The above discussions imply that either all the codewords in $O$ are positive or all the codewords in $O$ are negative. An orbit of the former type will be called a positive orbit, while an orbit of the latter type will be called a negative orbit.\\
  \indent If $O$ is a positive orbit of $\overline{\Omega}$ and $c_{a,b,c}\in O$ is a $0^+$-codeword, then for any $(\alpha,\beta)\in G$, $(\alpha,\beta)\cdot c_{a,b,c}$ is a $\beta^+$-codeword. In particular, for any $0\le s\le p-1$, there are $p-1$ $s^+$-codewords in $O$.\ If $c_{a,b,c}$ is an $s^+$-codeword, then
  $$
    w_H(c_{a,b,c})=\begin{cases}
      (p-1)(p^{m-1}-p^{\frac{m}{2}-1}), & \text{if $s=0$,}    \\
      (p-1)p^{m-1}+p^{\frac{m}{2}-1},   & \text{if $s\ne 0$.}
    \end{cases}
  $$
  Hence, there are $p-1$ codewords of Hamming weight $(p-1)(p^{m-1}-p^{\frac{m}{2}-1})$ and $(p-1)^2$ codewords of Hamming weight $(p-1)p^{m-1}+p^{\frac{m}{2}-1}$ in $O$. Similarly, if $O$ is a negative orbit of $\overline{\Omega}$ and $c_{a,b,c}\in O$ is a $0^-$-codeword, then for any $0\le s\le p-1$, there are $p-1$ $s^-$-codewords in $O$. Moreover, there are $p-1$ codewords of Hamming weight $(p-1)(p^{m-1}+p^{\frac{m}{2}-1})$ and $(p-1)^2$ codewords of Hamming weight $(p-1)p^{m-1}-p^{\frac{m}{2}-1}$ in $O$.\\
  \indent Assume that there are $k_+$ positive orbits and $k_-$ negative orbits in $\overline{\Omega}$. Then
  \begin{equation}\label{20230602equationtwo}
    k_++k_-=\frac{p^m(p^m-1)}{p-1}.
  \end{equation}
  To determine $k_+$ and $k_-$, we need to consider the weight distribution of $\overline{C_{\Pi}}$ first. By the above discussions, we know that $\overline{A}_i=0$ for $i\ne 0$, $(p-1)p^{m-1}$, $p^m$, $(p-1)p^{m-1}\pm p^{\frac{m}{2}-1}$ and $(p-1)(p^{m-1}\pm p^{\frac{m}{2}-1})$.\ Moreover, we have
  \begin{align*}
     & \overline{A}_0=1,                                         \\
     & \overline{A}_{(p-1)(p^{m-1}-p^{\frac{m}{2}-1})}=(p-1)k_+, \\
     & \overline{A}_{(p-1)p^{m-1}-p^{\frac{m}{2}-1}}=(p-1)^2k_-, \\
     & \overline{A}_{(p-1)p^{m-1}}=p^{m+1}-p,                    \\
     & \overline{A}_{(p-1)p^{m-1}+p^{\frac{m}{2}-1}}=(p-1)^2k_+, \\
     & \overline{A}_{(p-1)(p^{m-1}+p^{\frac{m}{2}-1})}=(p-1)k_-, \\
     & \overline{A}_{p^m}=p-1.
  \end{align*}
  \indent If $p=3$, then by \cite[Theorem 7]{carlet2005linear}, the minimum distance of the dual code $\overline{C_{\Pi}}^{\perp}$ of $\overline{C_{\Pi}}$ is $5$. By calculating the first four Pless power moments (see \cite[p.90]{pless1998handbook}), we have
  \begin{align*}
     & \sum\limits_{j=0}^{3^m}\overline{A}_j=3^{2m+1},                                           \\
     & \sum\limits_{j=0}^{3^m}j\overline{A}_j=2\cdot 3^{3m},                                     \\
     & \sum\limits_{j=0}^{3^m}j^2\overline{A}_j=2\cdot 3^{3m-1}(2\cdot 3^m+1),                   \\
     & \sum\limits_{j=0}^{3^m}j^3\overline{A}_j=2\cdot 3^{3m-2}(4\cdot 3^{2m}+2\cdot 3^{m+1}-1).
  \end{align*}
  This is a system of linear equations with variables
  $$\overline{A}_{2\cdot (3^{m-1}\pm 3^{\frac{m}{2}-1})},\quad \overline{A}_{2\cdot 3^{m-1}\pm 3^{\frac{m}{2}-1}}.$$
  Since the coefficient matrix of this system is a Vandermonde matrix, it has a unique solution, which is
  \begin{align*}
     & \overline{A}_{2\cdot (3^{m-1}\pm 3^{\frac{m}{2}-1})}=\frac{3^m(3^m-1)}{2}, \\
     & \overline{A}_{2\cdot 3^{m-1}\pm 3^{\frac{m}{2}-1}}=3^m(3^m-1).
  \end{align*}
  It follows that
  $$k_+=k_-=\frac{3^m(3^m-1)}{4}.$$
  \indent Now assume that the condition b) holds. Then the functions $\Pi_a:\mathbb{F}_q\rightarrow\mathbb{F}_p$ given by $x\mapsto{\rm{Tr}}_{\mathbb{F}_q/\mathbb{F}_p}(a\Pi(x))$ are weakly regular bent for all $a\in\mathbb{F}_q^*$. By Corollary \ref{20230603corolone}, we know that if $c_{a,0,0}\in\overline{\Omega}$ is a positive (resp., negative) codeword, then so is $c_{a,b,c}$ for any $b$, $c\in\mathbb{F}_q$. Hence, to prove that $k_+=k_-$, it suffices to show that half of the codewords $c_{a,0,0}$ ($a\in\mathbb{F}_q^*$) are positive while the other half are negative. By the proof of Theorem \ref{20230602themfour}, it suffices to show that half of the weakly regular bent functions $\Pi_a$ ($a\in\mathbb{F}_q^*$) have sign $+1$ while the other half have sign $-1$.\\
  \indent By Lemma \ref{20230602lemmatwo}, for any $a\in\mathbb{F}_q^*$, we have
  $$W_{\Pi_a}(0)=\epsilon_a\xi_p^{\Pi_a^*(0)}\sqrt{q},$$
  where $\epsilon_a\in\{\pm 1\}$ is the sign of $\Pi_a$. Note that
  \begin{align*}
    \sum\limits_{a\in\mathbb{F}_q^*}W_{\Pi_a}(0) & =\sum\limits_{a\in\mathbb{F}_q^*}\sum\limits_{x\in\mathbb{F}_q}\xi_p^{{\rm{Tr}}_{\mathbb{F}_q/\mathbb{F}_p}(a\Pi(x))} \\
                                                 & =\sum\limits_{x\in\mathbb{F}_q}\sum\limits_{a\in\mathbb{F}_q^*}\xi_p^{{\rm{Tr}}_{\mathbb{F}_q/\mathbb{F}_p}(a\Pi(x))} \\
                                                 & =\sum\limits_{\substack{x\in\mathbb{F}_q                                                                              \\\Pi(x)=0}}(q-1)+\sum\limits_{\substack{x\in\mathbb{F}_q\\\Pi(x)\ne 0}}(-1)\\
                                                 & =|\Pi^{-1}(0)|(q-1)-(q-|\Pi^{-1}(0)|)                                                                                 \\
                                                 & =q|\Pi^{-1}(0)|-q=0,
  \end{align*}
  which implies that
  \begin{align}\label{20230604equationtwo}
    \sum\limits_{a\in\mathbb{F}_q^*}\epsilon_a\xi_p^{\Pi_a^*(0)}=0.
  \end{align}
  The above equality can be reformulated as
  $$\sum\limits_{r=0}^{p-1}\big(\sum\limits_{\substack{a\in\mathbb{F}_q^*\\\Pi_a^*(0)=r}}\epsilon_a\big)\xi_p^r=0,$$
  which implies that $\xi_p$ is a root of the polynomial
  $$\sum\limits_{r=0}^{p-1}\big(\sum\limits_{\substack{a\in\mathbb{F}_q^*\\\Pi_a^*(0)=r}}\epsilon_a\big)X^r\in\mathbb{Z}[X].$$
  Since the minimal polynomial of $\xi_p$ over $\mathbb{Q}$ is $X^{p-1}+X^{p-2}+\cdots+1$, the sums
  $$S_r=\sum\limits_{\substack{a\in\mathbb{F}_q^*\\\Pi_a^*(0)=r}}\epsilon_a,\quad 0\le r\le p-1$$
  must be all the same. Since the map $\mathbb{F}_q^*\rightarrow\mathbb{F}_p$ given by $a\mapsto\Pi_a^*(0)$ is not surjective, we have $S_i=0$ for some $0\le i\le p-1$, which implies that $S_r=0$ for any $0\le r\le p-1$. However,
  \begin{align*}
        & \sum\limits_{r=0}^{p-1}S_r  =\sum\limits_{a\in\mathbb{F}_q^*}\epsilon_a           \\
    =\  & \#\{a\in\mathbb{F}_q^*:\ \epsilon_a=1\}-\#\{a\in\mathbb{F}_q^*:\ \epsilon_a=-1\}.
  \end{align*}
  Therefore, half of the weakly regular bent functions $\Pi_a$ ($a\in\mathbb{F}_q^*$) have sign $+1$ while the other half have sign $-1$. This completes the proof.
\end{proof}

\begin{remark}
  From the proof of Theorem \ref{202363theoremsix}, we can observe that
  when $m$ is odd, positive codewords and negative codewords in $\overline{\Omega}$ can be mutually transformed through the $G$-action, whereas when $m$ is even, the positive and negative codeword sets are both $G$-stable. This leads to the need for additional conditions to ensure an equal number of positive and negative codewords in $\overline{\Omega}$ when $m$ is even.
\end{remark}

\begin{remark}
  Assume that $\Pi$ is of Dembowski-Ostrom type or  Coulter-Matthews type. By \cite[Lemma 2]{feng2007value}, $\Pi^{-1}(0)=\{0\}$. By \cite[Lemma 3 ii)]{feng2007value}, $\Pi$ is weakly regular perfect nonlinear and the map $\mathbb{F}_q^*\rightarrow\mathbb{F}_p$ given by $a\mapsto\Pi_a^*(0)$ is the zero constant map. Therefore, all known perfect nonlinear functions $\Pi:\mathbb{F}_q\rightarrow\mathbb{F}_q$ satisfy the condition b).
\end{remark}

\begin{comment}

\begin{conjecture}
  Does every perfect nonlinear function $\Pi:\mathbb{F}_q\rightarrow\mathbb{F}_q$ satisfy the condition b)?
\end{conjecture}

\end{comment}

\begin{example}
  By Theorem \ref{202363theoremsix} and \cite{codetable}, the linear codes $\overline{C_{\Pi}}$ contain the following optimal codes:
  \begin{align*}
    [9,5,4;3], \quad[27,7,15;3],\quad [125,7,95;5],
  \end{align*}
  and the following best known codes:
  \begin{align*}
    [81,9,48;3],\quad[243,11,153;3].
  \end{align*}
\end{example}

\indent Finally, we determine the weight distribution of the linear code $C_{\Pi}$, which can also be done by considering the $G$-action and calculating the first few Pless power moments.

\begin{theorem}\label{20230601themthree}
  Let $\Pi:\mathbb{F}_{q}\rightarrow\mathbb{F}_{q}$ be a perfect nonlinear function with $\Pi(0)=0$. Assume that $\Pi(cx)\ne c\Pi(x)$ for any $x\in\mathbb{F}_q^*$ and $c\in\mathbb{F}_p\backslash\{0,1\}$.
  \begin{enumerate}
    \item If $m$ is odd, then $A_i=0$ except for the values
          \begin{align*}
             & A_0=1,                                                                               \\
             & A_{(p-1) p^{m-1}-p^{\frac{m-1}{2}}}=(p-1)(p^m-1)\frac{p^{m-1}+p^{\frac{m-1}{2}}}{2}, \\
             & A_{(p-1) p^{m-1}}=(p^{m-1}+1)(p^m-1),                                                \\
             & A_{(p-1) p^{m-1}+p^{\frac{m-1}{2}}}=(p-1)(p^m-1)\frac{p^{m-1}-p^{\frac{m-1}{2}}}{2}.
          \end{align*}
          In particular, the minimum distance of the code $C_{\Pi}$ is $(p-1)p^{m-1}-p^{\frac{m-1}{2}}$ and $C_{\Pi}$ is a $3$-weight code.
    \item If $m$ is even, then $A_i=0$ except for the values
          \begin{align*}
             & A_0=1,                                                                                            \\
             & A_{(p-1)(p^{m-1}-p^{\frac{m}{2}-1})}=\frac{p^m-1}{2}(p^{m-1}+p^{\frac{m}{2}}- p^{\frac{m}{2}-1}), \\
             & A_{(p-1)p^{m-1}-p^{\frac{m}{2}-1}}=\frac{p^m-1}{2}(p-1)(p^{m-1}+p^{\frac{m}{2}-1}),               \\
             & A_{(p-1) p^{m-1}}=p^m-1,                                                                          \\
             & A_{(p-1)p^{m-1}+p^{\frac{m}{2}-1}}=\frac{p^m-1}{2}(p-1)(p^{m-1}-p^{\frac{m}{2}-1}),               \\
             & A_{(p-1) (p^{m-1}+p^{\frac{m}{2}-1})}=\frac{p^m-1}{2}(p^{m-1}-p^{\frac{m}{2}}+p^{\frac{m}{2}-1}),
          \end{align*}
          assuming furthermore that one of the following conditions holds:
          \begin{enumerate}
            \item $p=3$,
            \item $\Pi^{-1}(0)=\{0\}$, $\Pi$ is weakly regular perfect nonlinear and the map $\mathbb{F}_q^*\rightarrow\mathbb{F}_p$ given by $a\mapsto\Pi_a^*(0)$ is not surjective, where $\Pi_a=f_{a,0,0}$ and $\Pi_a^*$ is the dual of $\Pi_a$.
          \end{enumerate}
          In particular, the minimum distance of $C_{\Pi}$ is $(p-1)(p^{m-1}-p^{\frac{m}{2}-1})$ and  $C_{\Pi}$ is a $5$-weight code.
  \end{enumerate}
\end{theorem}
\begin{proof}
  Note that $w_H(c_{a,b})=w_H(c_{a,b,0})$ for any $a$,  $b\in\mathbb{F}_{q}$. By \cite[Theorem 6]{carlet2005linear}, the minimum distance of the dual code $C_{\Pi}^{\perp}$ of $C_{\Pi}$ is at least $3$.\\
  \indent  We first treat the case where $m$ is odd. By Theorem \ref{202363theoremsix}, all the possible nonzero Hamming weights in $C_{\Pi}$ are $(p-1)p^{m-1}$ and $(p-1)p^{m-1}\pm p^{\frac{m-1}{2}}$. Then the assertion follows by calculating the first three Pless power moments, which has been done in {\cite[Theorem 2]{yuan2006weight}}.\\
  \indent Next we treat the case where $m$ is even. If $a=0$ and $b\ne 0$, then $w_H(c_{a,b,0})=(p-1)p^{m-1}$. There are $p^m-1$ such codewords. Put
  $$\Omega=\{c_{a,b,0}\in\overline{\Omega}\}.$$
  By Theorem \ref{202363theoremsix}, any codeword $c_{a,b,0}\in\Omega$ is an $s^+$-codeword or an $s^-$-codeword, where $0\le s\le p-1$. It is clear that there are $A_{(p-1)(p^{m-1}-p^{\frac{m}{2}-1})}$ $0^+$-codewords in $\Omega$, $A_{(p-1)(p^{m-1}+p^{\frac{m}{2}-1})}$ $0^-$-codewords in $\Omega$, $A_{(p-1)p^{m-1}+p^{\frac{m}{2}-1}}$ strictly positive codewords in $\Omega$ and $A_{(p-1)p^{m-1}-p^{\frac{m}{2}-1}}$ strictly negative codewords in $\Omega$. If $c_{a,b,0}\in\Omega$ is a positive codeword, then in $(\{1\}\times\mathbb{F}_p)\cdot c_{a,b,0}$, there is exactly one $s^+$-codeword for any $0\le s\le p-1$. Similar conclusion holds for negative codewords. Since $\overline{\Omega}=(\{1\}\times\mathbb{F}_{p})\cdot\Omega$, by Table I we have
  \begin{align}\label{20230601equationfour}
     & A_{(p-1)(p^{m-1}-p^{\frac{m}{2}-1})}+A_{(p-1)p^{m-1}+p^{\frac{m}{2}-1}}=\frac{p^m(p^m-1)}{2},       \\
     & A_{(p-1)(p^{m-1}+p^{\frac{m}{2}-1})}+A_{(p-1)p^{m-1}-p^{\frac{m}{2}-1}}=\frac{p^m(p^m-1)}{2}.\notag
  \end{align}
  By calculating the second and the third Pless power moments, we have
  \begin{align}\label{20230601equationthree}
     & \sum\limits_{j=0}^{p^m-1}jA_j=p^{2m-1}(p-1)(p^m-1),                                  \\
     & \sum\limits_{j=0}^{p^m-1}j^2A_j=p^{2m-2}(p-1)(p^m-1)\big(p+(p-1)(p^m-2)\big). \notag
  \end{align}
  Then (\ref{20230601equationfour}) and (\ref{20230601equationthree}) form a system of linear equations with variables
  $$A_{(p-1)p^{m-1}\pm p^{\frac{m}{2}-1}}\quad\mbox{and}\quad A_{(p-1)(p^{m-1}\pm p^{\frac{m}{2}-1})}.$$
  It turns out that this system has a unique solution, i.e., the one given in the theorem. This completes the proof.
\end{proof}

\begin{remark}
  If $\Pi$ is of  Coulter-Matthews type, then $\Pi$ is even, i.e., $\Pi(x)=\Pi(-x)$ for any $x\in\mathbb{F}_{3^m}$. Assume that there exists $x\in\mathbb{F}_q^*$ such that $\Pi(2x)=2\Pi(x)$, then
  $$\Pi(x)=\Pi(-x)=\Pi(2x)=2\Pi(x),$$
  which implies that $\Pi(x)=0$. However, by \cite[Lemma 2]{feng2007value}, we have $\Pi^{-1}(0)=\{0\}$, which is a contradiction. Therefore, $\Pi$ satisfies the conditions in Theorem \ref{20230601themthree}.
\end{remark}
\begin{remark}
  If $\Pi$ is of Dembowski-Ostrom type, then
  $$\Pi(x)=\sum\limits_{0\le i\le j\le m-1}a_{ij}x^{p^i+p^j}$$
  for some $a_{ij}\in\mathbb{F}_q$. If there exists $y\in\mathbb{F}_q^*$ and $c\in\mathbb{F}_p^*\backslash\{0,1\}$ such that $\Pi(cy)=c\Pi(y)$, then we have
  \begin{align*}
    c\Pi(y)=\Pi(cy) & =\sum\limits_{0\le i\le j\le m-1}a_{ij}(cy)^{p^i+p^j}             \\
                    & =\sum\limits_{0\le i\le j\le m-1}a_{ij}c^{p^i+p^j}y^{p^i+p^j}     \\
                    & = \sum\limits_{0\le i\le j\le m-1}a_{ij}c^{p^i}c^{p^j}y^{p^i+p^j} \\
                    & =\sum\limits_{0\le i\le j\le m-1}a_{ij}c^{2}y^{p^i+p^j}           \\
                    & = c^2\sum\limits_{0\le i\le j\le m-1}a_{ij}y^{p^i+p^j}=c^2\Pi(y).
  \end{align*}
  Since $\Pi$ is even and $\Pi(0)=0$, by \cite[Lemma 2]{feng2007value}, we have $\Pi(y)\ne 0$. It follows that $c^2=c$, which implies that $c=0$ or $1$. This is a contradiction. Therefore, $\Pi$ satisfies the conditions in Theorem \ref{20230601themthree}.
\end{remark}

\indent As a summary, in Theorem \ref{202363theoremsix} and Theorem \ref{20230601themthree}, we describe completely the weight distributions of the linear codes $C_{\Pi}$ and $\overline{C_{\Pi}}$ only using the assumption that $\Pi$ has perfect nonlinearity for all odd primes $p$ when $m$ is odd and give some mild additional conditions for similar conclusions to hold when $m$ is even. As shown in the remarks, these two theorems cover all the previous relevant results (e.g., \cite[Theorem 2]{li2008covering}, \cite[Theorem 1, 2]{li2008covering}, \cite[Theorem 4]{li2008properties}) as special cases.\\
\indent The approach employed in this paper, which studies the evolution of the codewords of a linear code under some group action is valuable and typical. It could potentially be extended to investigate other linear codes constructed from perfect nonlinear functions or bent functions (e.g., the linear codes considered in \cite{zheng2020subfield}).

\bibliographystyle{IEEEtranS}
\bstctlcite{IEEEexample:BSTcontrol}
\bibliography{references}{}

\end{document}